\documentclass[12pt]{amsart}%
\usepackage{amsmath}
\usepackage{amssymb}
\usepackage{amsthm}
\usepackage{amscd}
\usepackage[dvips]{graphicx}
\usepackage[mathscr]{eucal}
\newtheorem{Thm}{Theorem}[section]
\theoremstyle{definition}
\newtheorem{Theorem}[Thm]{Theorem}
\newtheorem{Lemma}[Thm]{Lemma}
\newtheorem{Corollary}[Thm]{Corollary}
\newtheorem{Proposition}[Thm]{Proposition}

\newtheorem{Example}[Thm]{Example}
\theoremstyle{remark}
\newtheorem{Remark}{Remark}
\font\sy=cmsy10

\font\ym=msbm10

\newcommand{\R}{\text{\ym R}}

\newcommand{\C}{\text{\ym C}}
\newcommand{\cB}{{\hbox{\sy B}}}

\newcommand{\cH}{{\hbox{\sy H}}}
\newcommand{\cK}{{\hbox{\sy K}}}

\title[Kakutani Dichotomy]
{Kakutani Dichotomy on Free States}

\author[Matsui]{Taku Matsui}
\address{Faculty of Mathematics, Kyushu University}
\email{matsui@math.kyushu-u.ac.jp}

\author[Yamagami]{Shigeru Yamagami*}
\address{Graduate School of Mathematics, Nagoya University}
\email{yamagami@math.nagoya-u.ac.jp}
\urladdr{http://www.math.nagoya-u.ac.jp/\~{}yamagami/}

\thanks{*Partially supported by KAKENHI(22540217)}

\begin{document}
\maketitle   
% \section{Kakutani's Dichotomy}
\begin{abstract}
Two quasi-free states on a CAR or CCR algebra are shown to generate quasi-equivalent 
representations unless they are disjoint. 
\end{abstract} 

\section{Introduction}
Kakutani's celebrated dichotomy theorem on infinite product measures opened a way to 
mathematical analysis in infinite dimensional phenomena. 
In classical  probability theory, lots of related results have been explored since then, whereas 
in quantum probability, 
this has been mostly done with relations to infinite tensor products of states of quantum algebras. 
Especially quasi-free states of so-called CAR algebras and CCR algebras were investigated much 
around 1970's from the view point of equivalence of representations and  
explicit criteria for their quaisi-equivalence are obtained 
in terms of Hilbert-Schmidt class operators. 

In this paper, we shall add a complement to this old subject by establishing 
dichotomies on quasi-free states:  
Given quasi-free states $\varphi$ and $\psi$ of a CAR or CCR algebra, one of the following 
alternatives occurs. 
\begin{enumerate}
\item
$\varphi$ and $\psi$ are quasi-equivalent. 
\item 
$\varphi$ and $\psi$ are disjoint. 
\end{enumerate}

In the case of CCR algebras, these alternatives are further related with non-vanishing or vanishing 
of transition probabilities between quasi-free states, which therefore inherits the same spirit with 
the original dichotomy due to S.~Kakutani.

\section{Preliminaries}
We shall freely use the standard terminologies in operator algebras and
the notations introduced in \cite{gmta} with some of basic ones repeated here 
for the reader's convenience. 
Given a C*-algebra $C$, $L^2(C)$ denotes 
the standard Hilbert space of the enveloping von Neumann algebra 
$C^{**}$ with the natural left and right actions of $C$ on $L^2(C)$. 
For a state $\varphi$ of $C$, the realizing vector in the positive cone of $L^2(C)$ is denoted by 
$\varphi^{1/2}$. The projection to the closed subspace $\overline{C\varphi^{1/2} C} \subset L^2(C)$ 
is then equal to the central support of $\varphi$, which is a projection in the center of $C^{**}$. 
As a consequence, two states $\varphi$ and $\psi$ produce quasi-equivalent GNS representations 
if and only if $\overline{C \varphi^{1/2}C} = \overline{C\psi^{1/2}C}$, whereas they are disjoint 
if and only if $\overline{C \varphi^{1/2}C} \perp \overline{C\psi^{1/2}C}$. 

In this framework, we have several possibilities for transition probability between states. 
Most known is the Uhlmann's one, which is the square of the so-called fidelity 
$\rho(\varphi,\psi)$ between states $\varphi$, $\psi$ (see \cite{AU} for further information). 
In our context of non-commutative $L^p$-theory 
(see \cite{Ha} among several approaches to the subject and also cf.~\cite{aamt}), 
$\rho(\varphi,\psi)$ is equal to the norm of the positive linear functional 
$|\varphi^{1/2} \psi^{1/2}| = \sqrt{\varphi^{1/2}\psi\varphi^{1/2}}$ in $C^*$ (\cite{R}). 
Another choice is $(\varphi^{1/2}|\psi^{1/2})$, which 
% looks like a transition amplitude in its appearance but 
is reduced to the ordinary transition probability for vector states on $\cB(\cH)$ and 
will play similar roles as Hellinger integrals did in the Kakutani's dichotomy theorem 
(\cite{Ka}). 
Thus its vanishing or non-vanishing is our main concern here and the fidelity can be 
equally well used 
for this purpose in view of inequalities 
$(\varphi^{1/2}|\psi^{1/2})^2 \leq \rho(\varphi, \psi)^2 \leq (\varphi^{1/2}|\psi^{1/2})$. 

For free states of quantum algebras, we know decisive results for 
the criterion of quasi-equivalence and the closed formula of transition probability. 
To explain these, we recall relevant definitions. 

Given a real Hilbert space $V$ with inner product $(x,y)$ 
($x,y \in V$), the CAR algebra is a unital C*-algebra $C(V)$ 
linearly generated by elements of $V$ with the relations 
\[
x^* = x, 
\quad 
xy + yx = (x,y)1, 
\qquad 
x, y \in V. 
\]
Likewise, given a real vector space $V$ and an alternating bilinear form $\sigma$ on $V$, 
the CCR C*-algebra is the C*-algebra $C(V,\sigma)$ generated universally by the symbols 
$\{ e^{ix}\}_{x \in V}$ with the relations 
\[
(e^{ix})^* = e^{-ix}, 
\quad 
e^{ix} e^{iy} = e^{-i\sigma(x,y)/2} e^{i(x+y)}, 
\qquad 
x, y \in V. 
\]
Remark that we allow $\sigma$ to be degenerate, whence our CCR C*-algebras may have non-trivial 
centers. 

Given a state $\varphi$ of a CAR algebra $C(V)$, the covariance operator $S$ on the complexified 
Hilbert space $V^\C$ is defined by 
$\varphi(x^*y) = (x,Sy)$, which turns out to be positive and satisfies the relation 
$S + \overline{S} = I$, where $\overline S$ is the complex conjugate of $S$ and $I$ denotes 
the identity operator. 
A state is said to be quasi-free and denoted by $\varphi_S$ 
if it vanishes on the odd part of $C(V)$ 
and satisfies the recursive relation 
\begin{multline*}
\varphi(x_1x_2 \dots x_{2n}) = \varphi(x_1x_2) \varphi(x_3x_4\dots x_{2n})\\ 
- \varphi(x_1x_3) \varphi(x_2x_4\dots x_{2n}) + \cdots 
+ \varphi(x_1x_{2n}) \varphi(x_2 \dots x_{2n}). 
\end{multline*}
If the recursive compuations are worked out completely, the Wick formula is obtained: 
\[
\varphi(x_1x_2 \dots x_{2n}) = 
\sum \pm \prod_{k=1}^n \varphi(x_{i_k}x_{j_k}), 
\]
where the summation is taken over all the way of pairings in $\{ 1, 2, \dots, 2n\}$ and 
$\pm$ is chosen according to the signature of the permutation sequence $(i_1,j_1, \dots, i_n,j_n)$. 

In the case of CCR C*-algebra, a state $\varphi$ is said to be quasi-free and denoted by 
$\varphi_S$ if 
\[
\varphi(e^{ix}) = e^{-S(x,x)/2}, 
\]
where $S$ is a positive sesqui-linear form on the complexfied vector space $V^\C$ and 
is referred to as the covariance form of $\varphi$. 
We know that a positive form $S$ on $V^\C$ is a covariance form if and only if 
\[
S(x,y) - \overline{S}(x,y) = i \sigma(x,y)
\]
for $x, y \in V^\C$. Here $\overline{S}(x,y) = \overline{S(\overline{x},\overline{y})}$ and 
$\sigma$ is sesqui-linearly extended to $V^\C$. 

Given a quas-free state $\varphi_S$, we write 
$L^2(S) = \overline{C \varphi_S^{1/2} C}$ with $C = C(V)$ (CAR case) or $C = C(V,\sigma)$ 
(CCR case). 
Notice here that the same letter is used to stand for a covariance operator or 
a covariance form according to the case of CAR or CCR. 

For quasi-free states, quasi-equivalence criteria were investigated by 
many researchers but let us just indicate 
\cite{Ar}, \cite{PS}, \cite{S}, \cite{SS} and \cite{V} among them. 
The following form is due to \cite{Ar} and \cite{AY}. 

\begin{Theorem}[Quasi-Equivalence Criteria]\label{QEQ}~ 
  \begin{enumerate}
\item 
Let $\varphi_S$ and $\varphi_T$ be quasi-free states of a CAR algebra with covariance operators $S$ 
and $T$. Then $\varphi_S$ and $\varphi_T$ are quasi-equivalent if and only if 
$\sqrt{S} - \sqrt{T}$ is in the Hilbert-Schmidt class. 
  \item 
Let $\varphi_S$ and $\varphi_T$ be quasi-free states of a CCR C*-algebra 
with covariance forms $S$ and $T$. 
Then $\varphi_S$ and $\varphi_T$ are quasi-equivalent if and only if 
$S + {\overline S} \cong T + {\overline T}$ as inner products and 
$\displaystyle 
\sqrt{\frac{S}{S + \overline S}} - \sqrt{\frac{T}{T + \overline T}}$ 
is in the Hilbert-Schmidt class. 
  \end{enumerate}
\end{Theorem} 

The following determinant formulas for transition probabilities are due to \cite{Ar} (CAR case) and 
\cite{gqfs} (CCR case). 

\begin{Theorem}[Transition Probability Formula]\label{TAF}~ 
  \begin{enumerate}
  \item 
Let $S$ and $T$ be covariance operators for a CAR algebra.  
Then 
\[
(\varphi_S^{1/2}|\varphi_T^{1/2})^4 = \det(MM^*), % \det(I - (P-Q)^2). 
\]
where $M = S^{1/2}T^{1/2} + (I-S)^{1/2}(I-T)^{1/2}$. 
\item 
Let $S$ and $T$ be covariance forms for a CCR C*-algebra. Then 
\[
(\varphi_S^{1/2}|\varphi_T^{1/2})^2 = \det 
\left( 
\frac{2\sqrt{AB}}{A+B} 
\right), 
\]
where positive forms $A$ and $B$ are defined by 
\[
2A = S + 2\sqrt{S\overline{S}} + \overline{S}, %(\sqrt{S} + \sqrt{\overline S})^2, 
\quad 
2B = T + 2\sqrt{T\overline{T}} + \overline{T}, %(\sqrt{T} + \sqrt{\overline T})^2
\]
and their geometric mean $\sqrt{AB}$ 
as well as $\sqrt{S\overline{S}}$ and 
$\sqrt{T\overline{T}}$ is in the sense of \cite{PW}. 
  \end{enumerate}
\end{Theorem}

\section{CAR Dichotomy}
Let $\epsilon$ be the parity automorphism 
of $C(V)$
% , $\Pi = \text{Ad}\,\epsilon$ be the associated unitary involution on $L^2(C(V))$, 
and define a bounded linear operator $\pi(\xi \oplus \eta)$ on $L^2(C(V))$ by 
\[
\pi(\xi \oplus \eta)\psi^{1/2} 
% = \xi\psi^{1/2} + (\Pi \psi^{1/2}) \eta 
= \xi\psi^{1/2} + (\psi\circ \epsilon)^{1/2}\eta.  
\]
Here $\xi, \eta \in V$ and $\psi$ is a state of $C(V)$. Then 
\[
\pi(\overline{\xi} \oplus - \overline{\eta}) \pi(\xi' \oplus \eta')
+ \pi(\xi' \oplus \eta')\pi(\overline{\xi} \oplus - \overline{\eta}) 
= (\xi|\xi') + (\eta|\eta')
\]
and $\pi$ is extended to a *-representation of $C(V\oplus iV)$, which is referred to as 
the \textbf{quadrate representation} of $C(V\oplus iV)$. 
Here $iV$ denotes the real part of $V^\C$ with respect to the conjugation given by 
$x \mapsto -\overline{x}$. 
Note that, if $\psi$ is an even state, i.e., $\psi \circ \epsilon = \psi$, 
\[
\pi(C(V\oplus iV))\psi^{1/2} 
= C(V)\psi^{1/2}C(V). 
\]
In particular, for a quasi-free state $\varphi_S$ 
of covariance operator $S$, $\pi$ leaves the closed central subspace 
$L^2(S) = \overline{C(V)\varphi_S^{1/2} C(V)}$ invariant. 
Let $\pi_S$ be the associated subrepresentation of $C(V \oplus iV)$. 

We define the \textbf{quadrature} of a state $\varphi$ of $C(V)$ 
to be a state $\Phi$ of $C(V\oplus iV)$ given by 
\[
\Phi(x) = (\varphi^{1/2}|\pi(x)\varphi^{1/2}), 
\quad 
x \in C(V\oplus iV).  
\] 
The following is well-known (see \cite{Ar} for example). 

\begin{Lemma}
The following conditions on a covariance operator $S$ are equivalent. 
\begin{enumerate}
\item 
$\ker S = \{0\}$.
\item 
$\ker(I-S) = \{ 0\}$. 
\item 
$S = (1+e^H)^{-1}$ with $H$ a self-adjoint operator on $V^\C$ satisfying 
$\overline{H} = -H$. 
\end{enumerate}
A covariance operator $S$ is said to be \textbf{non-degenerate} if it satisfies 
these equivalent conditions. 
\end{Lemma}

Given a covariance operator $S$ on $V^\C$, its quadrature is defined to be the projection 
\[
P = 
\begin{pmatrix}
S & \sqrt{S(I-S)}\\
\sqrt{S(I-S)} & I-S
\end{pmatrix}  
\]
on $V^\C\oplus V^\C$, which is a covariance operator for the real Hilbert space $V \oplus iV$. 

\begin{Proposition}\label{main}
% Given a covariance operator $C$ on $V$,  let $P$ be its quadrature. Then 
The quadrature of $\varphi_S$ is equal to the Fock state $\varphi_P$. 
In particular, the representation $\pi_S$ is irreducible. 
\end{Proposition}

\begin{proof} 
Recall that the Fock vacuum $\varphi_P^{1/2}$ is characterized by the vanishing property 
under the left multiplication of the range of $\overline{P}$. 
Since the range of $\overline{P}$ is equal to 
$\{ \sqrt{I-S}\zeta \oplus - \sqrt{S}\zeta; \zeta \in V^\C\}$, 
it suffices to show that 
\[
(\sqrt{I-S}\zeta) \varphi_S^{1/2} = \varphi_S^{1/2}(\sqrt{S}\zeta)\quad 
\text{for $\zeta \in V^\C$.}
\] 
If $S$ is non-degenerate, this follows from the fact that 
$\varphi_S$ is a KMS-state with repsect to 
the one-parameter automorphism group induced from the Bogoliubov transformations 
$\{ e^{itH}\}_{t \in \R}$ (see \cite[Example 5.3.24]{BR2} for example). 

To deal with the degenerate case, 
let $E$ be the projection to $\ker S(I-S)$ and write $(I-E) V^\C = W^\C$ 
with $W$ a closed real subspace of $V$. 
Let $\varphi_W$ (resp.~$\psi$) be the restriction of $\varphi_S$ to 
the C*-subalgebra $C(W) \subset C(V)$ (resp.~the C*-subalgebra $C(W^\perp) \subset C(V)$), 
which is a quasi-free state of the reduced covariance operator $S(I-E)$ (resp.~$SE$). 
Let $u$ be the unitary operator on the Fock space $\overline{C(W^\perp)\psi^{1/2}}$ defined by 
\[
u(\eta_1\cdots \eta_n\psi^{1/2}) = (-1)^n \eta_1\cdots \eta_n \psi^{1/2}\quad 
\text{for $\eta_1, \dots, \eta_n \in W^\perp$,}
\]
which implements the parity automorphism of $C(W^\perp)$. 

A representation $\theta$ of $C(V)$ on 
$\overline{C(W)\varphi_W^{1/2}} \otimes \overline{C(W^\perp) \psi^{1/2}}$ 
is then defined by the correspondance 
\[
\xi + \eta \mapsto \xi\otimes u + 1\otimes \eta,  
\quad 
\xi \in W,\ \eta \in W^\perp  
\]
on generators, where $\xi$ and $\eta$ on the right side denote operators by left multiplication. 
From $u \psi^{1/2} = \psi^{1/2}$ and the Wick formula, we have the equality 
\begin{multline*}
(\varphi_W^{1/2}\otimes \psi^{1/2}| 
(\xi_1\cdots \xi_m \otimes u^m\eta_1\cdots \eta_n) (\varphi_W^{1/2}\otimes \psi^{1/2}))\\ 
= \varphi_W(\xi_1\cdots \xi_m) \psi(\eta_1\cdots \eta_n) 
= \varphi_S(\xi_1\cdots \xi_m \eta_1\cdots \eta_n), 
\end{multline*}
which implies that 
$\xi_1\cdots \xi_m \eta_1 \cdots \eta_n \varphi_S^{1/2} 
\mapsto \phi(\xi_1\cdots \xi_m \eta_1\cdots \eta_n) 
(\varphi_W^{1/2} \otimes \psi^{1/2})$ gives rise to an isometry $U$. 
Since the operator $u$ is approximated by elements in $C(W^\perp)$ on 
$\overline{C(W^\perp)\psi^{1/2}}$ thanks to the irreducibility of representation, 
$U$ is in fact surjective and $\theta$ is extended to an isomorphism 
$C(V)'' \to C(W)'' \otimes \cB(\overline{C(W^\perp)\psi^{1/2}})$ of von Neumann algebras 
so that $\varphi_S = (\varphi_W\otimes \psi)\theta$, which in turn induces 
an isometric isomorphism
\[
\Theta: \overline{C(V) \varphi_S^{1/2} C(V)} 
\to
\overline{C(W) \varphi_W^{1/2} C(W)} \otimes 
\overline{C(W^\perp) \psi^{1/2} C(W^\perp)} 
\]
by the relation 
\[
\Theta(x\varphi_S^{1/2}x') = \theta(x) (\varphi_W^{1/2}\otimes \psi^{1/2}) \theta(x'), 
\quad 
x, x' \in C(V). 
\]

Now, for $\xi + \eta \in V^\C = (I-E)V^\C + EV^\C$, in view of 
$((I-S)\eta) \psi^{1/2} = 0 = \psi^{1/2} (S\eta)$, we see that 
\begin{align*}
\Theta(\varphi_S^{1/2}(\sqrt{S}(\xi + \eta)) 
&= (\varphi_W^{1/2} \otimes \psi^{1/2}) \theta(\sqrt{S}(\xi + \eta))\\ 
&= (\varphi_W^{1/2} \otimes \psi^{1/2}) (\sqrt{S}\xi\otimes u + 1\otimes \sqrt{S}\eta)\\ 
&= \varphi_W^{1/2}(\sqrt{S}\xi) \otimes \psi^{1/2} 
= (\sqrt{I-S}\xi)\varphi_W^{1/2} \otimes \psi^{1/2}\\
&= \theta(\sqrt{I-S}(\xi + \eta)) (\varphi_W^{1/2} \otimes \psi^{1/2})\\ 
&= \Theta(\sqrt{I-S}(\xi + \eta) \varphi_S^{1/2}). 
\end{align*}
% Then 
% $\xi \oplus \eta \in \overline{P}(1-(E\oplus E))(V^\C\oplus V^\C)$ implies  
% $\pi(\xi\oplus \eta)\varphi_C^{1/2} = 0$ by the previous KMS-state analysis in view of the above lemma. 
% On the other hand, $\xi \oplus \eta \in \overline{P}(E\oplus E)(V^\C\oplus V^\C)$ implies 
% $C\xi = 0$ and $(I-C)\eta = 0$, whence we have again 
% \[
% \pi(\xi \oplus \eta)\varphi_C^{1/2} = \xi\varphi_C^{1/2} + \varphi_C^{1/2}\eta = 0, 
% \]
% in view of 
% \[
% \varphi_C(\xi^*\xi) = (\xi|C\xi) = 0, 
% \quad  
% \varphi_C(\eta\eta^*) = (\eta^*|C\eta^*) = (\eta|(I-C)\eta) = 0. 
% \]
% Thus $\varphi_C^{1/2}$ is annihilated by $\pi(\xi\oplus \eta)$ 
% if $\xi \oplus \eta \in \overline{P}(V^\C \oplus V^\C)$; 
% $(\varphi_C^{1/2}|\pi(\cdot)\varphi_C^{1/2})$ is the Fock state $\varphi_P$. 
\end{proof}

\begin{Theorem}[Dichotomy]\label{dichotomy}
Let $S$ and $T$ be covariance operators for a CAR algebra $C(V)$ with $P$ and $Q$ their 
quadratures. Let $L^2(S) = \overline{C(V)\varphi_S^{1/2}C(V)}$ and similarly for $L^2(T)$. 
Then $L^2(S) \perp L^2(T)$ unless $L^2(S) = L^2(T)$. 
Moreover, we have 
\[
(\varphi_P^{1/2}|\varphi_Q^{1/2}) = (\varphi_S^{1/2}|\varphi_T^{1/2})^2. 
\] 
\end{Theorem}

\begin{proof}
Since $\pi$ is irreducible on both of 
$L^2(S)$ and $L^2(T)$, 
they are either unitarily equivalent or disjoint as representations of $C(V\oplus iV)$. 
% we have a dichotomy: 
% $\overline{C(V\oplus V)\varphi_S^{1/2}} \cap \overline{C(V\oplus V)\varphi_T^{1/2}} = \{ 0\}$ or 
% $\overline{C(V\oplus V)\varphi_S^{1/2}} = \overline{C(V\oplus V)\varphi_T^{1/2}}$. 
% From the equalities
% \[
% \overline{\pi(C(V\oplus V))\varphi_S^{1/2}} 
% = \overline{C(V)\varphi_S^{1/2}C(V)}, 
% \quad 
% \overline{\pi(C(V\oplus V))\varphi_T^{1/2}} 
% = \overline{C(V)\varphi_T^{1/2}C(V)}, 
% \] 
% the former implies $\overline{C(V\oplus V)\varphi_S^{1/2}} \perp \overline{C(V\oplus V)\varphi_T^{1/2}}$. 
Let $z_S$ be the projection to 
$\overline{\pi(C(V\oplus iV)) \varphi_S^{1/2}} = L^2(S)$ and similarly for $z_T$. 
Since $z_S$ is in the commutant of the right representation of $C(V)$ on $L^2(S)$, 
it is approximated by the left multiplication of $C(V)$, i.e., by elements in $\pi(C(V\oplus 0))$. 
Thus, if a unitary $U: L^2(S) \to L^2(T)$ intertwines $\pi$, then 
\[
U(\xi) = U(z_S\xi) = z_SU(\xi), 
\quad 
\xi \in L^2(S)
\]
shows that $z_S = z_T$, i.e., $L^2(S) = L^2(T)$. 

Otherwise, by the irreducibility of 
$\pi(C(V \oplus iV))$ on both $L^2(S)$ and $L^2(T)$, 
\[
\pi(C(V \oplus iV)) \varphi_S^{1/2} \perp 
\pi(C(V \oplus iV)) \varphi_T^{1/2}, 
\]
i.e., $z_S \perp z_T$. 
Then $\Phi_S^{1/2}$ and $\Phi_T^{1/2}$ belong to inequivalent irreducible components of 
a representatiopn of $C(V \oplus iV)$, whence they are orthogonal.  
% Now assume that $z_S\perp z_T$. Then irreducible representations $\pi$ on 
% $L^2(S)$ and $L^2(T)$ are disjoint with $\Phi_S$ and $\Phi_T$ in different components, whence
% \[
% (\Phi_S^{1/2}|\Phi_T^{1/2}) = 0 = (\varphi_S^{1/2}|\varphi_T^{1/2})^2. 
% \]
% Otherwise, both of $\Phi_S$ and $\Phi_T$ are vector states of $\varphi_S^{1/2}$ and 
% $\varphi_T^{1/2}$ through a single irreducible representation of $\pi$ on 
% $L^2(S) = L^2(T)$, whence 

In either case, we have 
\[
(\Phi_S^{1/2}|\Phi_T^{1/2}) = \text{trace}\Bigl( 
|\varphi_S^{1/2})(\varphi_S^{1/2}|\, 
|\varphi_T^{1/2})(\varphi_T^{1/2}|
\Bigr)
= (\varphi_S^{1/2}|\varphi_T^{1/2})^2. 
\]
\end{proof}

\begin{Remark}
For factorial states, this kind of dichotomy is an immediate consequence of Schur's lemma. 
In the case of CAR, non-factorial quasi-free states are known to be
decomposed into two pure states and we can work explicitly with these exceptional cases 
to get the dichotomy. 
\end{Remark}

\begin{Remark}
Let $C_0(V)$ be the even part of $C(V)$, which is the fixed point subalgebra by 
the parity automorphism. Let $S$ be a covariance operator 
such that $S(I-S)$ is in the trace class and $\ker(2S - I)$ is even-dimensional. 
Then we can find Fock states $\varphi_j$ ($j=1,2$) of $C(V)$ such that 
$\varphi_j$ is quasi-equivalent to $\varphi_S$ ($j=1,2$),
the restrictions $\psi_j = \varphi_j|_{C_0(V)}$ are inequivalent pure states of $C_0(V)$, 
and $\varphi_S|_{C_0(V)} = (\psi_1 + \psi_2)/2$. 
Thus $\varphi_S|_{C_0(V)}$ is neither quasi-equivalent nor disjoint to 
both of $\psi_j$. See \cite{M} for more information. 
\end{Remark}

\section{CCR Dichotomy}
A state $\varphi$ of a C*-algebra $C$ is said to be 
\textbf{standard} if $\overline{C \varphi^{1/2}} = \overline{\varphi^{1/2}C}$. 

\begin{Example}\label{nonstandard}
Let $\varphi = \varphi_1 \otimes \varphi_2$ be a product state 
on $C = C_1 \otimes C_2$ with $\varphi_1$  a pure state of $C_1$. 
Then 
$\overline{C_1\varphi_1^{1/2} C_1} 
\cong \overline{C_1 \varphi_1^{1/2}} \otimes \overline{\varphi_1^{1/2} C_1}$ 
and $\varphi$ is not standard if 
$\dim \overline{C_1\varphi_1^{1/2}} = \dim \overline{\varphi_1^{1/2} C_1} \geq 2$. 
\end{Example} 

\begin{Lemma}~ 
\begin{enumerate}
\item
Let $S$ be the covariance form of a quasi-free state $\varphi$ of a CCR C*-algebra $C(V,\sigma)$. 
Then $\varphi$ is standard if and only if the kernel of the ratio operator 
$\frac{S}{S + {\overline S}}$ on $V_S^\C$ is trivial. Here $V_S^\C$ denotes the Hilbert space 
induced from $S + \overline{S}$ on $V^\C$. 
\item
Let $S$ be the covariance operator of a quasi-free state $\varphi$ of a CAR algebra $C(V)$. 
Then $\varphi$ is standard if and only if $\ker S = \{ 0\}$. 
\end{enumerate}
\end{Lemma} 

\begin{proof}
Sufficiency: 
Since $\varphi$ is a KMS-state, this follows from \cite[Lemma~2.3]{gmta}. 

Necessity: If a covariance form has a non-trivial kernel, 
the associated quasi-free state is factored through a pure state and therefore it is not standard  
in view of Example~\ref{nonstandard}. 
\end{proof}

\begin{Corollary}\label{perturbation} 
Let $\varphi$ be the quasi-free state of a covariance form $S$ and 
suppose that the kernel of $S/(S + {\overline S})$ (CCR case) or the kernel of $S$ (CAR case) is 
separable. 
Then we can find a standard quasi-free state $\varphi'$ such that 
$\varphi$ and $\varphi'$ are quasi-equivalent.   
\end{Corollary} 

\begin{Lemma}\label{lp}
For standard states $\varphi$ and $\psi$, $(\varphi^{1/2}|\psi^{1/2}) = 0$ if and only if 
$\overline{C\varphi^{1/2}C} \perp \overline{C \psi^{1/2} C}$, i.e., $\varphi$ and $\psi$ 
are disjoint. 
\end{Lemma} 

\begin{proof}
This is a consequence of Schwarz inequality and the tracial property of the evaluation map 
in non-commutative $L^p$-theory: For $a, b \in C$, 
\begin{align*}
|(a\varphi^{1/2}|b\psi^{1/2})| &= 
|(\psi^{1/4} b^*a \varphi^{1/4}| \psi^{1/4}\varphi^{1/4})|\\ 
&\leq \| \varphi^{1/4} a^*b \psi^{1/4}\|_2\,\|\psi^{1/4}\varphi^{1/4})\|_2\\ 
&= \| \varphi^{1/4} a^*b \psi^{1/4}\|_2\, \sqrt{(\varphi^{1/2}|\psi^{1/2})} = 0. 
\end{align*}
\end{proof} 

\begin{Lemma}\label{key}
For quasi-free states $\varphi$ and $\psi$, $(\varphi^{1/2}|\psi^{1/2}) > 0$ implies 
their quasi-equivalence, i.e., 
$\overline{C \varphi^{1/2} C} = \overline{C \psi^{1/2} C}$ 
($C = C(V)$ or $C(V,\sigma)$). 
\end{Lemma} 

\begin{proof}
% By the explicit formula for the transition probability $(\varphi^{1/2}|\psi^{1/2})$, 
% the strict positivity necessitates fulfilling the Hilbert-Schmidt criterion 
% on quasi-equivalence. 
% [CAR]  
We shall deal only with the case of CCR and the easier CAR case is omitted. 
In view of the determinant formula (Theorem~\ref{TAF}), 
we first rewrite the condition that $(\varphi_S^{1/2}|\varphi_T^{1/2}) > 0$. 
The equivalence (i.e., mutual dominations) of $S + {\overline S}$ and 
$T + {\overline T}$ is necessary, which is assumed in the following. 
Because of 
\[
S + {\overline S} \leq 2A \leq 2(S+{\overline S}), 
\quad 
T + {\overline T} \leq 2B \leq 2(T+{\overline T}), 
\]
these as well as $A + B$ are equivalent. In particular, the ratio operator 
\[
\frac{\sqrt{AB}}{A+B}
\]
is invertible and the transition probability does not vanish if and only if 
\[
I - 2\frac{\sqrt{AB}}{A+B} = \left( \sqrt{\frac{A}{A+B}} - \sqrt{\frac{B}{A+B}} \right)^2 
\]
is in the trace-class. In view of 
\[
\left( \sqrt{\frac{A}{A+B}} + \sqrt{\frac{B}{A+B}} \right)^2 
\left( \sqrt{\frac{A}{A+B}} - \sqrt{\frac{B}{A+B}} \right)^2 
= \left( \frac{A}{A+B} - \frac{B}{A+B} \right)^2 
\]
and the invertibility of $\frac{A}{A+B}$ and $\frac{B}{A+B}$, the condition is equivalent to requiring 
that 
\[
\frac{2A}{A+B} - \frac{2B}{A+B} = 
\frac{S+{\overline S} + 2\sqrt{S\overline{S}}}{A+B} -  
\frac{T+{\overline T} + 2\sqrt{T\overline{T}}}{A+B}
\] 
is in the Hilbert-Schmidt class. The last condition is equivalent to 
the quasi-equivalence of $\varphi_S$ and $\varphi_T$ 
by \cite[Theorem, Proposition 6.6, Proposition~9.1]{AY}. 

In this way, we have proved that $\varphi_S$ and $\varphi_T$ are quasi-equivalent if %and only if 
$(\varphi_S^{1/2}|\varphi_T^{1/2}) > 0$. 
\end{proof} 

% \begin{proof}[Proof of Dichotomy] 
% Let $\varphi$ and $\psi$ be quasi-free states with 
% $\varphi'$ and $\psi'$ perturbed states as in Corollary. 
% If $\varphi$ and $\psi$ are not quasi-equivalent, the same holds for $\varphi'$ and $\psi'$, 
% whence $((\varphi')^{1/2}|(\psi')^{1/2}) = 0$ by Lemma C. 
% Since $\varphi'$ and $\psi'$ are standard, this implies the disjointness of $\varphi'$ and $\psi'$ by 
% Lemma B, proving the disjointness of $\varphi$ and $\psi$. 
% \end{proof}

In the case of CAR algebras, the converse of Lemma~3.7 is false. 

\begin{Proposition}
Let $S$ and $T$ be covariance operators on $V^\C$ with 
$P$ and $Q$ their quadratures on $(V \oplus iV)^\C = V^\C \oplus V^\C$. 
Assume that $\varphi_S$ and $\varphi_T$ are quasi-equivalent. 
Then $(\varphi_S^{1/2}| \varphi_T^{1/2}) = 0$ if and only if $P \wedge (I-Q) \not= 0$. 
\end{Proposition}

% \begin{Example}
% If a covariance operator $C$ is given by a projection, then 
% both of $\varphi_C$ and $\varphi_{1-C}$ are unitarily equivalent pure states but 
% \end{Example}

In the case of CCR C*-algebras, however, the transition probability is already sensetive to 
the dichotomy:  

% We shall prove the following dichotomy for quasi-free states on CCR algebras. 

\begin{Theorem}[Dichotomy]
Let $(V,\sigma)$ be a presymplectic vector space ($\sigma$ being an alternating bilinear form on $V$) 
and $S$, $T$ be covariance forms  
with the associated quasi-free states 
$\varphi_S$, $\varphi_T$. Then the following conditions are equivalent. 
\begin{enumerate}
\item 
Two states $\varphi_S$ and $\varphi_T$ are quasi-equivalent. 
\item 
The transition probability $(\varphi_S^{1/2}|\varphi_T^{1/2})$ is strictly positive. 
\item 
Positive forms $(\sqrt{S} + \sqrt{\overline S})^2$ and $(\sqrt{T} + \sqrt{\overline T})^2$ on $V^\C$ are Hilbert-Schmidt 
equivalent. 
\end{enumerate}
Otherwise, $\varphi_S$ and $\varphi_T$ are disjoint. 
\end{Theorem}

\begin{proof}
In view of Lemma~\ref{key}, it suffices to show that 
the condition $(\varphi_S^{1/2}|\varphi_T^{1/2}) = 0$ implies the disjointness of 
$\varphi_S$ and $\varphi_T$. 

By replacing $V^\C$ with $V_{A+B}^\C$, we may assume that $A+B$ is non-degenerate and complete 
on $V^\C$. 
If 
\[
\ker \left(  
\sqrt{\frac{A}{A+B}} \sqrt{\frac{B}{A+B}} 
\right) 
% = \ker \left( 
% \frac{A}{A+B} \frac{B}{A+B}
% \right) 
\]
is not trivial, we can find $0 \not= h \in V^\C$ such that $(A/(A+B))h = 0$ or $(B/(A+B))h = 0$. 
We may assume that the former is the case. Since the operator $A/(A+B)$ is self-conjugate, 
we can further assume that $h = \overline{h}$, i.e., $h \in V$. Now the condition  $A(h,h) = 0$ implies 
$S(h,h) = 0 = \overline{S}(h,h)$, whence $S(h,v) = 0 = \overline{S}(h,v)$ and 
$\sigma(h,v) = 0$ for any $v \in V$. Thus $\{ e^{ith} \}_{t \in \R}$ is in the center of $C^*(V,\sigma)$. 
Since $h \not= 0$ with respect to the inner product $A+B$, we see $B(h,h) \not= 0$ and therefore 
$T(h,h) \not= 0$. 

We now compare the spectral decomposition of $\{ e^{ith} \}$ when represented by left multiplication: 
On the subspace $\overline{C^*(V,\sigma) \varphi_S^{1/2} C^*(V,\sigma)}$, it is represented by the identity 
operator, whereas on the subspace $\cH = \overline{C^*(V,\sigma) \varphi_T^{1/2} C^*(V,\sigma)}$ it is isomorphic 
to a direct sum of the multiplication operator $\{ e^{itx} \}$ on $L^2(\R,\gamma)$ 
($\gamma$ being a Gaussian measure);  
\[
(e^{itx}\xi)(\tau) = e^{it\tau}\xi(\tau) 
\quad
\text{with} 
\quad
\xi \in L^2(\R,\gamma). 
\]
% \[
% e^{ith}\varphi_T^{1/2} = \int_\R^\oplus e^{itx - x^2/4} \xi(x)\, dx
% \quad 
% \text{in}
% \quad 
% \cH = \int_{\R}^\oplus \cH(x)\, dx. 
% \]
Thus $\varphi_S$ and $\varphi_T$ are disjoint. 

We now assume that the kernel of $\sqrt{AB}/(A+B)$ is trivial. Then, by the determinant formula for 
the transition probability, $(\varphi_S^{1/2}|\varphi_T^{1/2}) = 0$ implies that the bounded operator 
\[
\left( 
\sqrt{\frac{A}{A+B}} - \sqrt{\frac{B}{A+B}} 
\right)^2
\]
is not in the trace class. 
In particular, we can find a sequence $\{ h_j \}_{j \geq 1}$ of $(A+B)$-orthonormal vectors in $V^\C$ such that 
\[
\sum_j (A+B)(h_j, 
\left( 
\sqrt{\frac{A}{A+B}} - \sqrt{\frac{B}{A+B}} 
\right)^2
h_j) = +\infty. 
\]

Let $M$ be the set of monomials of $S/(A+B)$, 
$T/(A+B)$, $\overline{S}/(A+B)$ and $\overline{T}/(A+B)$. 
Let $W^\C$ be the closed subspace spanned by 
$\{ Mh_j, M\overline{h_j} \}_{j \geq 1}$. 
Since $M$ is countable, $W^\C$ is seperable as a hilbertian vector space. 
Clearly $W^\C$ is invariant under four generators of $M$. 
In view of $i\sigma = S - {\overline S}$, $W^\C$ is also invariant under $\sigma/(A+B)$. 
Since $M$ is closed under taking conjugate, so is $W^\C$, which justifies the notation, 
i.e., $W$ denotes the real part of $W^\C$. 
Let $W^\perp$ be the orthogonal complement of $W$ relative to $A+B$ so that 
$(V,\sigma) = (W,\sigma) \oplus (W^\perp,\sigma)$ with $S$ and $T$ diagonally decomposed. 
Let $S_W$ and $T_W$ 
be the reduced covariance forms. Then 
\[
\frac{\sqrt{A_WB_W}}{A_W + B_W}
\]
is the restriction of $\sqrt{AB}/(A+B)$ to the subspace $W^\C$ and 
\begin{align*}
\text{trace}
&\left( 
\sqrt{\frac{A_W}{A_W+B_W}} - \sqrt{\frac{B_W}{A_W+B_W}} 
\right)^2\\
&= \text{trace}_{W^\C} 
\left( 
\sqrt{\frac{A}{A+B}} - \sqrt{\frac{B}{A+B}} 
\right)^2\\
&\geq 
\sum_j (A+B)(h_j, 
\left( 
\sqrt{\frac{A}{A+B}} - \sqrt{\frac{B}{A+B}} 
\right)^2
h_j) = +\infty, 
\end{align*}
which means that $\varphi_{S_W}$ and $\varphi_{T_W}$ are not quasi-equivalent 
by Theorem~\ref{QEQ}. 

Now the obvious identification 
\begin{multline*}
\overline{C^*(V,\sigma) \varphi_S^{1/2} C^*(V,\sigma)} 
=\\  
\overline{C^*(W,\sigma) \varphi_{S_W}^{1/2} C^*(V,\sigma)} 
\otimes 
\overline{C^*(W^\perp,\sigma) \varphi_{S_{W^\perp}}^{1/2} C^*(W^\perp, \sigma)}
\end{multline*}
reveals that the disjointness of $\varphi_S$ and $\varphi_T$ follows from that of 
$\varphi_{S_W}$ and $\varphi_{T_W}$. Thus the problem is reduced to the case $W = V$ 
so that $V$ is separable relative to the inner product $A+B$ 
(so we omit the suffix $W$) and that $\varphi_S$ and $\varphi_T$ are not 
quasi-equivalent. 
We can then find standard covariance forms $S'$ and $T'$ such that 
$\varphi_S$ and $\varphi_{S'}$ (resp.~$\varphi_T$ and $\varphi_{T'}$) are quasi-equivalent 
by Corollary~\ref{perturbation}. 
% and 
% \begin{align*}
% \overline{C^*(V,\sigma)\varphi_{S'}^{1/2}} &= \overline{C^*(V,\sigma)\varphi_{S'}^{1/2} C^*(V,\sigma)},\\ 
% \overline{C^*(V,\sigma)\varphi_{T'}^{1/2}} &= \overline{C^*(V,\sigma)\varphi_{T'}^{1/2} C^*(V,\sigma)}. 
% \end{align*}

Since $\varphi_S$ and $\varphi_T$ are not quasi-equivalent, the same holds for $\varphi_{S'}$ and 
$\varphi_{T'}$, which implies the disjointness of $\varphi_{S'}$ and $\varphi_{T'}$ by 
Lemma~\ref{lp}. 
Thus $L^2(S) = L^2(S')$ is orthogonal to $L^2(T) = L^2(T')$, 
proving the disjointness of $\varphi_S$ and $\varphi_T$. 
% is equivalent to $(\varphi_{S'}^{1/2}|\varphi_{T'}^{1/2}) = 0$. 
% Now for $x \in V$
% \begin{align*}
% |(\varphi_{S'}^{1/2}| e^{ix} \varphi_{T'}^{1/2})|
% &= |\langle \varphi_{T'}^{1/4} \varphi_{S'}^{1/4}  \varphi_{S'}^{1/4} e^{ix} \varphi_{T'}^{1/4} \rangle|\\ 
% &\leq \| \varphi_{T'}^{1/4} \varphi_{S'}^{1/4}\|\, 
% \| \varphi_{S'}^{1/4} e^{ix} \varphi_{T'}^{1/4} \|\\
% &= (\varphi_{S'}^{1/2} | \varphi_{T'}^{1/2})^{1/2}\,  
% \| \varphi_{S'}^{1/4} e^{ix} \varphi_{T'}^{1/4} \| = 0
% \end{align*}
% shows the orthogonality of 
% $C^*(V,\sigma)\varphi_{S'}^{1/2}$ and $C^*(V,\sigma)\varphi_{T'}^{1/2}$, proving the disjointness of 
% $\varphi_{S'}$ and $\varphi_{T'}$. 
\end{proof}

% We first reduce the problem to a seperable cases. 
% Let $S + {\overline S} + T + {\overline T}$ be a reference inner product on $V$ and take the completion of $V$. 
% The relevant sesquilinear forms $S$, $T$ and $\sigma$ are then expressed by bounded hermitian operators, 
% which are denoted by $[S]$, $[T]$ and $[\sigma]$ respectively.

% Now transfinite induction is used to get a family of seperable closed subspaces $\{ V_i\}$ of $V$ such that
% (i) $V_i \perp V_j$ ($i \not= j$), (ii) $V_i^\C$ is invariant under $M$ and 
% (iii) the algebraic sum $\sum_i V_i$ is dense in $V$. 

% Let $\sigma_i$, $S_i$ and $T_i$ be the restrictions to the subspace $V_i^\C$. Then the obvious identification 

% Otherwise, $A$ and $B$ are equivalent and we have a direct sum decomposition of $V$ into seperable subspaces. 
% The condition $(\varphi_S^{1/2}|\varphi_T^{1/2}) = 0$ 
% is equivalent to the divergence of the trace of 
% \[
% \left( 
% \sqrt{\frac{A}{A+B}} - \sqrt{\frac{B}{A+B}} 
% \right)^2
% \]
% and then at least one of seperable direct summands of $V$ inherits this property. 
% From the obvious identification
% \[
% \overline{(A\otimes B)(\varphi\otimes \psi)^{1/2} (A\otimes B)} = 
% \overline{A\varphi^{1/2}A}\otimes \overline{B \psi^{1/2} B}, 
% \]
% the disjointness is reduced to the case that $A$ and $B$ give a seperable hilbertian topology. 


\begin{thebibliography}{20}

\bibitem{AU} 
P.M.~Alberti and A.~Uhlmann, 
On Bures distance and *-algebraic transition probability 
between inner derived positive linear forms over W*-algebras,
\textit{Acta Appl.~Math.}, 60(2000), 1--37. 

\bibitem{Ar}
H.~Araki, On quasifree states of CAR and Bogoliubov automorphisms, 
Publ.~RIMS, 6(1970/1971), 385--442. 

% \bibitem{A} 
% H.~Araki, On quasifree states of the canonical 
% commutation relations II, \textit{Publ.~RIMS}, 7(1972), 121--152. 

\bibitem{AY} 
H.~Araki and S.~Yamagami, 
On quasi-equivalence of quasifree states of 
the canonical commutation relations, 
\textit{Publ.~RIMS}, 18(1982), 703--758.

% \bibitem{B} 
% F.A.~Berezin, \textit{The Method of Second Quantization}, 
% Academic Press, 1966. 

% \bibitem{RS}
% M.~Reed and B.~Simon, \textit{Functional Analysis I}, Academic Press, 1980.

% \bibitem{Bo}
% V.I.~Bogachev, \textit{Gaussian Measures}, 
% Amer.~Math.~Soc., 1998. 

% \bibitem{BR1} 
% O.~Bratteli and D.W.~Robinson, 
% \textit{Operator Algebras and Quantum Statistical Mechanics I}, 
% Springer-Verlag, 1979. 

\bibitem{BR2} 
O.~Bratteli and D.W.~Robinson, 
\textit{Operator Algebras and Quantum Statistical Mechanics II}, 
Springer-Verlag, 1979. 

% \bibitem{Fe} 
% J.~Feldman, 
% Equivalence and perpendicularity of gaussian processes, 
% \textit{Pacific J.~Math.}, 8(1958), 699--708. 

% \bibitem{M} 
% Paul-Andr\'e Meyer, 
% \textit{Quantum Probability for Probabilists}, LNM~1538, 
% Springer-Verlag, 1993. 

\bibitem{Ha}
U.~Haagerup, $L^p$-spaces associated with an arbitrary von Neumann algebra, 
Colloques internationaux du CNRS, No.~274, 1977. 

% \bibitem{H} 
% A.S.~Holevo, Quasifree states of the C* algebra 
% of CCR. II, \textit{Theor.~Math.~Physics}, 
% 6(1971), 103--107.

% \bibitem{J} 
% S.~Janson, \textit{Gaussian Hilbert Spaces}, 
% Cambridge University Press, 1997.

\bibitem{Ka}
S.~Kakutani, On equivalence of infinite product measures, 
\textit{Ann.~Math.}, 49(1948), 214--224. 

% \bibitem{K} 
% H.-H.~Kuo, 
% \textit{Gaussian Measures in Banach Spaces}, 
% LNM~463, Springer-Verlag, 1975. 

% \bibitem{MV} 
% J.~Manuceau and A.~Verbeure, 
% \textit{Commun.math.Phys.}, 
% 9(1968), 

% \bibitem{MSTV} 
% J.~Manuceau, M.~Sirugue, D.~Testard 
% and A.~Verbeure, 
% The smallest C*-algebra for 
% canonical commutations relations, 
% \textit{Commun.~Math.~Phys.}, 32(1973), 231--243. 

\bibitem{M}
T.~Matsui, 
Factoriality and quasi-equivalence of quasifree states for $Z_2$ and $\text{U}(1)$  
invariant CAR algebras, 
\textit{Rev.~Roumaine Math.~Pures Appl.}, 32(1987), 693-700. 

% \bibitem{MS} 
% T.~Matsui and Y.~Shimada, 
% On qusifree representations of infinite dimensional 
% symplectic group, 
% \textit{J.~Funct.~Analysis}, 215(2004), 67--102. 

% \bibitem{P} 
% R.J.~Plymen, Automorphic group representations: 
% The hyperfinite $\text{II}_1$ factor and the Weyl algebra, 
% in \textit{Alg\`ebres d'Op\'erateurs LNM~725}, Springer-Verlag, 1979. 

\bibitem{PS}
R.T.~Power and E.~St\o rmer, Free states of the canonical anticommutation relations, 
\textit{Commun.~Math.~Phys.}, 16(1970), 1--33. 
% \bibitem{R} 
% D.W.~Robinson, 
% Commun.math.Phys.,
% 1(1965), 159--
%The first appearance of CCR C*-algebra

\bibitem{PW} 
W.~Pusz and S.L.~Woronowicz, 
Functional calculus for sesquilinear froms and 
the purification map, 
\textit{Rep.~Math.~Phys.}, 
8(1975), 159--170. 

\bibitem{R} 
G.A.~Raggio, 
Comparison of Uhlmann's transition probability 
with the one induced by the natural cone of 
von Neumann algebras in standard form, 
\textit{Lett.~Math.~Phys.}, 6(1982), 233--236.

% \bibitem{Sc} 
% H.~Scutaru, Transition probabilities between 
% quasifree states, 
% arXiv:quant-ph/9908061.

% \bibitem{Se} 
% I.E.~Segal, 
% Distributions in Hilbert space and canonical 
% systems of operators, \textit{Trans.~Amer.~Math.~Soc.}, 
% 88(1958), 12--41. %JSTOR

\bibitem{S} 
D.~Shale, Linear symmetries of free Boson fields, 
\textit{Trans.~Amer.~Math.~Soc.}, 103(1962), 149--167. 

\bibitem{SS}
D.~Shale and W.F.~Stinespring, States on the Clifford algebra, 
\textit{Ann.~Math.}, 80(1964), 365--381. 

% \bibitem{ST} 
% T.~Shirai and Y.~Takahashi, 

% \bibitem{Si} 
% B.~Simon, \textit{The $P(\phi)_2$ Euclidean (Quantum) 
% Field Theory}, Princeton University Press, 1974. 

% \bibitem{Sl} 
% J.~Slawny, On factor representations 
% and the C*-algebra of canonical commutation 
% relations, \textit{Commun.~Math.~Phys.}, 
% 24(1972), 151--170. 

% \bibitem{Sp} 
% H.~Spohn, 
% Fermionic point process, 

% \bibitem{VV} 
% A.~Van Daele and A.~Verbeure, 
% Unitary equivalence of Fock representations 
% on the Weyl algebra, 
% \textit{Commun.~Math.~Phys.}, 
% 20(1971), 268--278. 

\bibitem{V} 
A.~Van Daele, 
Quas-equivalence of quasi-free states 
on the Weyl algebra, 
\textit{Commun.~Math.~Phys.}, 
21(1971), 171--191. 

% \bibitem{W1} 
% S.L.~Woronowicz, 
% On the purification of factor states, 
% \textit{Commun.~Math.~Phys.}, 
% 28(1972), 221--235. 

% \bibitem{W2} 
% S.L.~Woronowicz, 
% On the purification map, 
% \textit{Commun.~Math.~Phys.},
% 30(1973), 55--67. 

\bibitem{aamt}
S.~Yamagami, Algebraic aspects in modular theory, 
\textit{Publ.~RIMS}, 28(1992), 1075--1106. 

\bibitem{gmta}
S.~Yamagami, Geometric mean of states and transition amplitudes, 
\textit{Lett.~Math.~Phys.}, 84(2008), 123--137. 

\bibitem{gqfs}
S.~Yamagami, Geometry of quasi-free states of CCR-algebras, 
\textit{Int.~J.~Math.}, 21(2010), 875--913. 
\end{thebibliography}
\end{document}